\documentclass[a4paper, conference, 10pt] {IEEEtran}

\usepackage{amsmath,amssymb}
\usepackage{amsthm}
\usepackage{algcompatible}
\usepackage{algorithm}
\usepackage{algpseudocode}
\usepackage[pdftex]{graphicx}
\usepackage{tikz}
\definecolor{cof}{RGB}{219,144,71}
\definecolor{pur}{RGB}{186,146,162}
\definecolor{greeo}{RGB}{91,173,69}
\definecolor{greet}{RGB}{52,111,72}
\usetikzlibrary{patterns}
\usetikzlibrary{arrows}
\usepackage{pgfplots}
\usepackage[a4paper,top=1.9cm,inner=1.3cm,outer=1.3cm,bottom=2.78cm]{geometry}
\usepackage{multirow}
\usepackage{hhline}
\usepackage{bbm}
\usepackage{subcaption}
\makeatletter
\newcommand*{\rom}[1]{\expandafter\@slowromancap\romannumeral #1@}
\newcommand{\Mod}[1]{\ (\mathrm{mod}\ #1)}
\makeatother
\begin{document}
%
\title{A Distributed Scheduling Algorithm to Provide Quality-of-Service in Multihop Wireless Networks}

\author{\IEEEauthorblockN{Ashok Krishnan K.S. and Vinod Sharma}
\IEEEauthorblockA{Dept. of ECE, Indian Institute of Science, Bangalore, India\\
Email: \{ashok, vinod\}@ece.iisc.ernet.in
}}

\maketitle

\begin{abstract}
Control of multihop Wireless networks in a distributed manner while providing end-to-end delay requirements  for different flows, is a challenging problem. Using the notions of Draining Time and Discrete Review from the theory of fluid limits of queues, an algorithm that meets delay requirements  to various flows in a network is constructed. The algorithm involves an optimization which is implemented in a cyclic distributed manner across nodes by using the technique of iterative gradient ascent, with minimal information exchange between nodes. The algorithm uses time varying weights to give priority to flows. The  performance of the algorithm is studied in a network  with interference modelled by independent sets.
\end{abstract}

%
\IEEEpeerreviewmaketitle

\section{Introduction and Literature Review}
A multihop wireless network consists of nodes communicating data to each other through time varying, stochastic wireless channels. A network controller has to make decisions about the routing of flows, scheduling of links and the power control at each node, subject to various constraints. The controller may be a single node, making a centralized decision after accessing all the relevant information. Alternatively, the control process may be localized, with nodes making decisions about themselves in a distributed manner. While distributed algorithms are attractive from an implementation perspective, they may not always catch up with the centralized algorithms in terms of performance. In many general network scenarios, however, the centralized control problem may itself be intractable, and there is a need for suboptimal algorithms that are easily implementable \cite{Ji}.

Providing Quality-of-Service (QoS) is a central theme in network literature. Flows may demand different kinds of QoS depending on the application which generated them. Some applications may require a guarantee on the end-to-end mean delay, whereas others, such as live streaming, may require a \emph{hard} guarantee on the deadline. Some other applications may ask for a minimum bandwidth to be guaranteed at all times. In the large queue length regime, one approach to provide mean delays is to translate these requirements in terms of \emph{effective bandwidth} and \emph{effective delay} from Large Deviations theory, and obtain solutions in the physical layer; in \cite{she2016energy}, the authors use this technique for a K-user downlink scenario. Such techniques, however, cannot be applied easily in the multihop context, owing to the complex coupling between the queues, which makes it difficult to have a simplified one-to-one translation between delay requirements and control actions\cite{LauSurvey}.

Another approach to network control is to consider \emph{throughput optimal} algorithms \cite{Tassiu} which use backpressure. Such algorithms stabilize the network for any arrival rate inside the \emph{network capacity region}. However, backpressure based algorithms may not have good delay performance, especially under light loads \cite{Cui,sharma2007opportunistic}. In \cite{neely2009intelligent}, an algorithm is proposed that improves the energy-delay trade-off in a queue by intelligent dropping of packets.In \cite{stai2016performance} the authors use a weighted backpressure scheme to provide improvements in performance metrics such as delay, over traditional backpressure systems. Using Markov Decision Processes (MDPs) \cite{puterman2014markov} has been another approach to provide QoS, in the single hop as well as multihop context \cite{Kumar}. In general settings, however, MDPs are not easy to handle owing to the huge size of the state space. Control based on Lyapunov Optimization is quite popular in the multihop network setting \cite{GeorgBook}. In \cite{SatVS} the authors study the problem of minimizing power while simultaneously providing hard deadline guarantees in a wireless network. In \cite{ashok2016distributed} the authors devise a randomized algorithm which provides targeted mean delay and hard deadline for flows in a multihop setting under the SINR model.

One way to construct tractable models of networks is to use the notion of fluid limits. A comprehensive treatment of the theory and techniques used in fluid control is available in \cite{Meyn08}. The idea is to establish a suitable scaling under which the network converges to a simpler, deterministic fluid network, in which the various processes can be modelled by systems of ordinary differential equations (ODEs). One then obtains control policies at the level of these ODEs, and translates them to the actual network setting. A number of results are available that relate the performance of the fluid control to the control of the actual stochastic system \cite{Meyn03}. A continuous control policy can be built from the fluid model by means of the technique of discrete review \cite{Bau02,Marg00}. Here, the network is reviewed at certain points in time, and control decisions, as well as the next review instant, are calculated based on the system state at the review instants. A scheduling algorithm based on draining time was proposed in \cite{SubLei}. The authors consider a network without interference constraints, and are able to obtain stability results for simple models. They also note that obtaining stability results in the general case may be difficult. In \cite{Bert15}, the authors develop a robust fluid model, which adds stochastic variability to the deterministic fluid process, and develop a polynomial time algorithm to solve the network control problem.
 
Incremental Gradient methods \cite{Ber10} have been used in neural networks and other areas, especially for implementing distributed optimization. Here, in order to optimize a separable sum of functions, the gradient of each constituent function is taken iteratively, instead of calculating the gradient of the sum function. This leads to a separable iterative process that leads to the optimal point in the limit.
Our main contributions in this paper are summarised below.
\begin{itemize}
	\item We propose an optimization problem  motivated by the draining time of the fluid model associated with the network. A draining time based scheduling algorithm is considered in \cite{SubLei}, but there is no interference, and the function being optimized is  different. In \cite{Kumar} the authors provide end to end hard deadline guarantees by solving the dual problem of an appropriate MDP. However, their model has unreliable links but no interference constraints.
	\item The control variables that appear in the optimization are given time varying weights, in order to give priority to flows whose delay requirements have not been met.
	\item The optimization is to be solved at review instants, and the control variables obtained at the beginning of a review instant are used till the next review instant. This makes it less computationally intensive than algorithms that require computations to be done at each time slot, such as in \cite{stai2016performance} or \cite{ashok2016distributed}; \cite{ashok2016distributed} also uses a different interference model. Discrete review is used in works such as \cite{Marg00}, but for open queuing networks; moreover, the implementation is centralized and they do not consider delay deadlines.
	\item We use iterative gradient ascent to obtain a distributed algorithm in order to solve the optimization problem. This can be implemented easily in a cyclic manner, with message passing between the nodes after each step. We also show how the projection step involved in the optimization can be done by messaging between links that share a node.
	\item   The algorithm only requires whether the  delay constraint has been met at the flow destination at each review instant, apart from the local state information. We do not need to compute metrics over paths, as in \cite{SatVS}. Also, the algorithm in \cite{SatVS} is not fully distributed. 
\end{itemize} 
The rest of this paper is organized as follows. In Section II, we describe the system model, and provide the corresponding Fluid Model, as well as formulate the optimization problem for our QoS problem. In Section III we develop a distributed algorithm to solve this problem. We provide the simulation results in Section IV, followed by the conclusion in Section V.

\section{System Model}

We consider a multihop network (Fig. \ref{fig1}), given by a graph $G=(V,E)$ where $V=\{1,2,..,N\}$ is the set of vertices and $E$, the set of links on $V$. We assume a slotted system, with the discrete time index denoted by $t\in \{0,1,2,...\}$. We have directional links, with  link $(i,j)$ from node $i$ to node $j$ having a time varying channel gain $\gamma_{ij}(t)$ at time $t$. At each node $i$, $A_i^c(t)$ denotes the i.i.d process of exogenous arrival of packets destined to node $c$, with mean arrival rate  $\lambda_i^c=\mathbb{E}[A_i^c(t)]$. All traffic in the network with the same destination $c$ is called  \emph{flow} $c$; the set of all flows is denoted by $F$. Each flow has a fixed set of routes to follow to its destination. At each node there are queues, with $Q_i^c(t)$ denoting the queue length at node $i$ corresponding to flow $c \in F$. The queues evolve as
\begin{align}
Q_i^f(t)= Q_i^f(0)+\sum_{s=0}^{t-1}(A_i^f(s)+\sum_{k\neq i}S_{ki}^f(s)-\sum_{j \neq i}S_{ij}^f(s)) \label{actualQueue}
\end{align}
where $S_{mn}^f(s)$ denotes the number of packets of flow $f$ that are transmitted from node $m$ to node $n$ in time slot $s$.
 Each node transmits at a fixed power $p$. The rate of transmission between node $i$ and node $j$ is $\mu_{ij}(t)=f(p,\gamma(t))$ where $f$ is some achievable rate function.
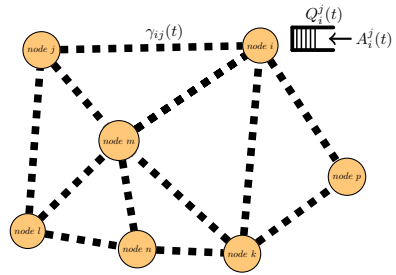
\begin{figure}
	\centering
	\setlength{\unitlength}{1cm}
	\thicklines
	\begin{tikzpicture}[scale=0.6, transform shape]		
	\node[draw,shape=circle, fill={rgb:orange,1;yellow,2;pink,5}, scale=0.6, transform shape] (v1) at (4.7,0.5) {$node\ k$};
	\node[draw,shape=circle, fill={rgb:orange,1;yellow,2;pink,5}, scale=0.6, transform shape] (v2) at (2.4,0.6) {$node\ n$};
	\node[draw,shape=circle, fill={rgb:orange,1;yellow,2;pink,5}, scale=0.6, transform shape] (v3) at (2.0,3.0) {$node\ m$};
	\node[draw,shape=circle, fill={rgb:orange,1;yellow,2;pink,5}, scale=0.6, transform shape] (v4) at (0.3,5) {$node\ j$};
	\node[draw,shape=circle, fill={rgb:orange,1;yellow,2;pink,5}, scale=0.6, transform shape] (v5) at (5.1,5.1) {$node\ i$};
	\node[draw,shape=circle, fill={rgb:orange,1;yellow,2;pink,5}, scale=0.6, transform shape] (v6) at (7,2.2) {$node\ p$};
	\node[draw,shape=circle, fill={rgb:orange,1;yellow,2;pink,5}, scale=0.6, transform shape] (v10) at (0,1.0) {$node\ l$};
	\node (v7) at (7.6,5.25) {$A_i^j(t)$};
	\node (v8) at (6.5,5.8) {$Q_i^j(t)$};
	\node (v9) at (3,5.4) {$\gamma_{ij}(t)$};			
	\draw[line width=1mm, dashed] (v2) -- (v1)
	(v4) -- (v5)			
	(v3) -- (v5)
	(v3) -- (v2)
	(v1) -- (v5)
	(v1) -- (v3)
	(v4) -- (v3)
	(v3) -- (v5)
	(v5) -- (v6)
	(v10) -- (v2)
	(v10) -- (v4)
	(v10) -- (v3)
	(v1) -- (v6);
	\draw[line width=0.5mm, line cap=round](5.8,5)--(6.7,5);
	\draw[line width=0.5mm, line cap=round](5.8,5.5)--(6.7,5.5);
	\draw[line width=0.5mm, line cap=round](5.8,5)--(5.8,5.5);
	\draw[line width=0.2mm](5.9,5)--(5.9,5.5);
	\draw[line width=0.2mm](6.0,5)--(6.0,5.5);
	\draw[line width=0.2mm](6.1,5)--(6.1,5.5);
	\draw[line width=0.2mm](6.2,5)--(6.2,5.5);
	\draw[line width=0.2mm](6.3,5)--(6.3,5.5);
	\draw[thick,->] (7.1,5.25) -- (6.6,5.25);
	\end{tikzpicture}
	\caption{A simplified depiction of a Wireless Multihop Network}
	\label{fig1}
\end{figure}
We will assume that the links are sorted into $M$ \emph{interference sets}\ $I_1, I_2, \dotsc, I_M$. At any time, only one link from an interference set can be active. A link may belong to multiple interference sets. In this work we will assume that any two links which share a common node will fall in the same interference set. 
\subsection{Fluid Model}
The \emph{fluid model} is an ODE approximation to the evolution of the queue. Under appropriate scaling, the queue evolution in equation (\ref{actualQueue}) converges\cite{Meyn08} to the \emph{fluid equation} given by
\begin{align*}
	\frac{d^+}{dt}q_i^f(t;x_i^f) = \lambda_i^f +\sum_{k\neq i} \zeta_{ki}^f(t)\mu_{ki}-\sum_{j\neq i} \zeta_{ij}^f(t)\mu_{ij}
\end{align*}
where $\frac{d^+}{dt}$ denotes the \emph{right derivative}, which is the derivative limit taken from the right, and  $q_i^f(0;x_i^f)=x_i^f$ is the initial condition, and $\zeta_{ij}^f$ corresponds to a fluid flow from node $i$ to node $j$, of flow $f$. This models the queue as a continuous time, deterministic process with continuous arrival and departure processes.  For the fluid model, the \emph{draining time} of a queue is defined to be the time to empty that queue. Consider the evolution of a fluid queue which has initial state $x$ and has  arrivals at rate $\lambda$ and service at rate $\mu(>\lambda)$, given by
\begin{align*}
q(t;x)=\max(x-(\mu-\lambda)t,0).
\end{align*}
The draining time in this case is given by
\begin{align*}
\tau=\frac{x}{\mu-\lambda}.
\end{align*}
 Draining time captures in some sense the delay associated with a flow. It is the time that an arrival at time $t=0$ would have to wait before it gets served, assuming first-in-first-out service discipline. In a multihop network, the draining time will be given by the smallest $\tau$ that solves the equation
 \begin{align*}
 x_i^f+\lambda_i^f\tau +\sum_{k\neq i}\mu_{ki}\int_0^{\tau} \zeta_{ki}^f(t)dt-\sum_{j\neq i}\mu_{ij}\int_0^{\tau} \zeta_{ij}^f(t)dt=0.
 \end{align*}
 Assuming $\zeta_{ij}^f(t)=\zeta_{ij}^f$ for all $t$, we obtain
 \begin{align*}
 \tau=\frac{x_i^f}{\sum_{j\neq i}\mu_{ij}\zeta_{ij}^f-\lambda_i^f-\sum_{k\neq i}\mu_{ki}\zeta_{ki}^f}.
 \end{align*}
 Calculating this requires knowledge of arrival rates as well as the scheduling decisions of other nodes. Hence, this term is not easy to calculate locally at a queue $q_i^f$. Therefore, we define a \emph{pseudo draining time}, given by
\begin{align*}
	D_i^f=\frac{x_i^f}{\sum_{j\neq i}\mu_{ij}\zeta_{ij}^f}.
\end{align*}
This $D_i^f$ is a lower bound to $\tau$. The draining time of the queue $q_i^f$ would have been $D_i^f$ if the queue had no inflow, and was serving at constant rate $\sum_{j\neq i}\mu_{ij}\zeta_{ij}^f$.  Consider the optimization
\begin{align*}
	\max\sum_{i,f}\frac{w_i^f}{D_i^f},
\end{align*}
where $w_i^f$ is a weight corresponding to flow $f$ on node $i$. Choosing $w_i^f=\theta^f(x_i^f)^2$, we obtain the problem
\begin{align*}
	\max\sum_{i,f}\theta^f x_i^f\sum_{j\neq i}\zeta_{ij}^f\mu_{ij}.
\end{align*}
Defining $\mu_{ii}=0$ for all $i$, we can rewrite this as
\begin{gather}
	\max\sum_{i,j,f}\theta^f x_i^f\zeta_{ij}^f\mu_{ij}, \label{optFun}\\
	s.t\ 0\leq\zeta_{ij}:=\sum_{f\in F}\zeta_{ij}^f\leq 1\ \forall ij, \label{optCon1}\\
	0\leq\zeta_{ij}+\zeta_{kl}\leq 1,\ \forall (i,j),(k,l)\in I_m, \forall m. \label{optCon2}
\end{gather}
where the first constraint corresponds to the fact that only one flow can be scheduled across a link, and the second constraint corresponds to interference constraints on the links. This is also the standard weighted-rate maximization problem \cite{WeeraBook}, with the weight given to rate $\mu_{ij}$ being $\sum_{f}\theta^f x_i^f\zeta_{ij}^f$.

We will use the technique of discrete review (See \cite{Marg00}, \cite{Meyn08} for discussions).  This involves an increasing sequence of times $0\leq T_1<T_2\dotsc$, at which we make control decisions for the network, based on the optimization problem (\ref{optFun})-(\ref{optCon2}). We solve the fluid problem, and obtain scheduling variables corresponding to those fluid variables, at every review instant.  In the time frame $[T_i,T_{i+1}]$, we will assume that the channel gain is fixed (slow-fading), but drawn as an \emph{i.i.d} sequence from a distribution. Each node transmits at power $P$. The rate $\mu_{ij}=\log(1+\frac{\gamma_{ij}P}{\sigma^2})$. Consider a packet of flow $f$ which arrives at node $i$ at the beginning of a review period. Such a packet observes a backlog of $x_i^f$ in its queue. The total service allocated to flow $f$ over link $(i,j)$ in that period is $\zeta_{ij}^f(T_{i+1}-T_i)$ where $\zeta_{ij}^f\leq1$.
The times are chosen as $T_{i+1}-T_i=a_1\log (1+a_2\sum_{i,f}Q_i^f(T_i))$. We will also be using \emph{safety stocks} \cite{Meyn08} to ensure that there is no starvation of resources in the network. These are thresholds below which each queue is not allowed to go.
\subsection{Providing Quality-of-Service}
We will be solving the optimization problem defined by equations (\ref{optFun})-(\ref{optCon2}) at every discrete review instant. In order to take care of QoS constraints, we will let $\theta^f$ vary dynamically. At a review instant, and flow $f_1$ requires its mean delay to be less than or equal to $d_1$. At the destination node of $f_1$, we estimate empirically the mean delay of flow $f_1$ up to the last review period, and if it is greater than $d_1$, we set $\theta^f=\hat{\theta}>1$. Otherwise, $\theta^f=1$. Thus the control variables corresponding to the flows that require QoS obtain higher weight in the optimization problem. For a hard deadline guarantee flow, we have two parameters, the hard deadline and the drop ratio, which is the percentage of packets we are willing to allow with delays larger than the hard deadline. At every review instant, at the destination of that flow, we check whether the percentage of packets that have arrived with delays larger than the deadline, exceeds the drop ratio. If yes, we set the corresponding $\theta^f=\hat{\theta}$. 
In the next section we will provide a distributed algorithm for the optimization problem defined by equations (\ref{optFun})-(\ref{optCon2}).
\section{Distributed Optimization}
Define the set of all \emph{link-flow} pairs $((i,j),f)$ as $\mathcal{K}$. For any $k\in\mathcal{K}$, there exists a link $(i(k),j(k))$ and a flow $f(k)$. A \emph{schedule} is a vector $\textbf{s}$ of length $|\mathcal{K}|$, with each element $\textbf{s}(k)$ corresponding to the fraction of time that link $(i(k),j(k))$ transmits flow $f(k)$. Define the \emph{feasible set} $\mathcal{S}$ to be the set of schedules that satisfy constraints (\ref{optCon1}) and (\ref{optCon2}); however, we remove the positivity constraint. Note that this changes the search space, but does not change the optimal value or the optimal point, since the quantity being maximized is a weighted sum of $\zeta_{ij}^f$ with positive weights. The set $\mathcal{S}$ will be a convex polytope, owing to the fact that it is generated by linear inequalities, and will be a closed subset of $\mathbb{R}^{|\mathcal{K}|}$.

We can rewrite  equations (\ref{optFun}) through (\ref{optCon2}) as
\begin{align}
\max_{\textbf{s}\in\mathcal{S}}\sum_{k\in\mathcal{K}}f_k(\textbf{s}) \label{DiffOpt}
\end{align}
with $f_k(\textbf{s})=w_k\mu_k\textbf{s}(k)$, $w_k=\theta^{f(k)} x_{i(k)}^{f(k)}$, $\mu_k=\mu_{i(k)j(k)}$ and $\textbf{s}(k)=\zeta_{i(k)j(k)}^{f(k)}$. 
\subsection{Incremental Gradient Ascent}
In order to optimize (\ref{DiffOpt}), we will use the incremental gradient method \cite{Ber10}. This involves the iteration
\begin{align*}
\textbf{s}_{j+1}=\Pi_{\mathcal{S}}(\textbf{s}_j+\alpha_j\nabla f_{k_j}(\textbf{s}_j)),
\end{align*}
with $k_j=j\ \text{modulo}\ |\mathcal{K}|+1$, and $\Pi_{\mathcal{S}}$ denotes projection onto the  set $\mathcal{S}$. Let $\textbf{v}(r)$  denote a vector which is one only at its $r$th index and zero elsewhere. We can write
\begin{align*}
\nabla f_{k_j}(\textbf{s}_j)=w_{k_j}\mu_{k_j}\textbf{v}(k_j).
\end{align*}
Hence we may write the gradient ascent equation as
\begin{align}
\textbf{s}^{'}_{j}=\textbf{s}_j+\alpha_j w_{k_j}\mu_{k_j}\textbf{v}(k_j). \label{DistriEq}
\end{align}
Following this, we do the projection, to obtain the new point
\begin{align}
\textbf{s}_{j+1}=\Pi_{\mathcal{S}}(\textbf{s}^{'}_{j}). \label{DistriEqP}
\end{align}
\subsection{Projection}
 Since  interference exists between any two links that share a common node, any update of the optimization variables at a link will affect only those links which share a node with it. The constraint set $\mathcal{S}$ is a  polytope, which is defined by the intersection of half-spaces $\{\mathcal{H}_i\}_{i=1}^M$, where each half-space $\mathcal{H}_i$ is characterized by an equation of the form
\begin{align*}
\langle \textbf{s},\boldsymbol{\nu}^i\rangle\leq \beta_i.
\end{align*}
 where $\boldsymbol{\nu}^i$ is the unit normal vector to the plane, with $||\boldsymbol{\nu}^i||_2=1$. For example, the interference constraint
 \begin{align*}
 s_1+s_2+s_4\leq 1
 \end{align*}
 can be represented by
 \begin{align*}
 \langle \textbf{s},\boldsymbol{\nu}^j\rangle\leq \beta_j,
 \end{align*}
 where $\boldsymbol{\nu}^j=\frac{1}{\sqrt{3}}\sum_{n=1,2,4}\textbf{v}(n)$, and $\beta_j=\frac{1}{\sqrt{3}}$.  Due to the nature of our constraints, $\boldsymbol{\nu}^i$ will have only non-negative components. Each of these half-spaces corresponds to exactly one interference constraint.
\begin{figure}
	\begin{subfigure}[b]{0.24\textwidth}
		\centering
		\resizebox{\linewidth}{!}{
			\begin{tikzpicture}
			\draw (0,2) -- (3,2);
			\draw[dashed] (3,2) -- (5,2);
			\draw (3,2) -- (5,0);
			\draw (0,-1) -- (5,0);
			\draw[dashed] (6,-1) -- (5,0);
			\draw[dashed] (3,2) -- (2,3);
			\draw[dashed] (5,0) -- (7,0.4);
			\fill[gray!20,opacity=0.6] (0,2) -- (3,2) -- (5,0) -- (0,-1) -- cycle;
			\node (v1) at (1.8,0.8) {$\mathcal{S}$};
			\node (v2) at (6.3,-1) {$\mathcal{H}_v$};
			\node (v3) at (5.3,2) {$\mathcal{H}_u$};
			\node (v4) at (7.3,0.4) {$\mathcal{H}_w$};
			\node (v5) at (1.6,2.5) {$A$};
			\node (v6) at (3.6,2.5) {$B$};
			\node (v6) at (5.0,1.2) {$C$};
			\node (v7) at (5.9,-0.3) {$D$};
			\node (v8) at (4.4,-0.7) {$E$};
			\node (v9) at (3.8,0.6) {$\textbf{s}$};
			\node (v10) at (5.8,1.4) {$\textbf{s}^{'}$};
			\node (v11) at (4.7,0.3) {$\textbf{s}^{''}$};
			\draw[dotted] (3.8,0.3) -- (5.8,1.0);
			\draw[dotted] (5.8,1.0) -- (4.9,0.1);
			\foreach \Point in {(3.8,0.3),(5.8,1.0),(4.9,0.1)}{
				\node at \Point {\textbullet};
			}
			\end{tikzpicture}
		}
		\caption{Single Step Projection}
		\label{figProj}
	\end{subfigure}
	\begin{subfigure}[b]{0.24\textwidth}
		\centering
		\resizebox{\linewidth}{!}{
			\begin{tikzpicture}
			\draw (0,2) -- (3,2);
			\draw[dashed] (3,2) -- (5,2);
			\draw (3,2) -- (5,0);
			\draw (0,-1) -- (5,0);
			\draw[dashed] (6,-1) -- (5,0);
			\draw[dashed] (3,2) -- (2,3);
			\draw[dashed] (5,0) -- (7,0.4);
			\fill[gray!20,opacity=0.6] (0,2) -- (3,2) -- (5,0) -- (0,-1) -- cycle;
			\node (v1) at (1.8,0.8) {$\mathcal{S}$};
			\node (v2) at (6.3,-1.4) {$\mathcal{H}_v$};
			\node (v3) at (5.3,2) {$\mathcal{H}_u$};
			\node (v4) at (7.3,0.4) {$\mathcal{H}_w$};
			\node (v5) at (1.6,2.5) {$A$};
			\node (v6) at (3.6,2.5) {$B$};
			\node (v6) at (5.0,1.2) {$C$};
			\node (v7) at (7.2,-0.3) {$D$};
			\node (v8) at (4.0,-0.7) {$E$};
			\node (v10) at (6.6,-0.1) {$\textbf{s}^{'}$};
			\node (v10) at (5.5,-1.0) {$\textbf{s}_1$};
			\node (v10) at (5.5,0.5) {$\textbf{s}_2$};
			\node (v10) at (4.8,-0.35) {$\textbf{s}_3$};
			\draw[dotted] (6.3,-0.45) -- (5.8,-0.8);
			\draw[dotted] (5.8,-0.8) -- (5.5,0.1);
			\draw[dotted] (5.5,0.1) -- (5.2,-0.25);
			\foreach \Point in {(6.3,-0.45),(5.8,-0.8),(5.5,0.1),(5.2,-0.25)}{
				\node at \Point {\textbullet};
			}
			\end{tikzpicture}
		}
		\caption{Iterative Projection}   
		\label{figProj2}
	\end{subfigure}
	\caption{Projection} 
	\label{fig:twoSubs}
\end{figure}
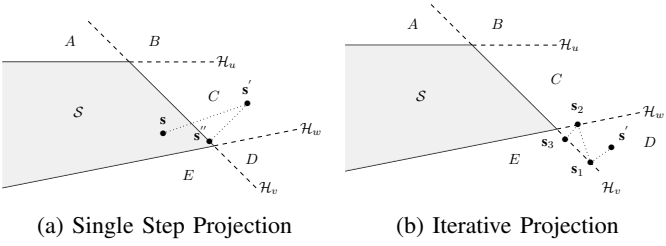\

In  step (\ref{DistriEq}), we update one component $\textbf{s}(k)$ of  $\textbf{s}$, which corresponds to a link flow pair $(i(k),j(k)),f(k)$. There are two half-space constraints, corresponding to links connected to $i(k)$ and  $j(k)$; let these be $\mathcal{H}_v$ and $\mathcal{H}_w$.  Fig. \ref{figProj} shows the case in which  equation (\ref{DistriEq}) violates $\mathcal{H}_v$ but not $\mathcal{H}_w$, moving from  $\mathcal{S}$ to  $C$. The  perpendicular projection  on the boundary of $\mathcal{H}_v$ is $\textbf{s}^{''}$. After (\ref{DistriEq}), if the point  moves to  $D$, projection can be done repeatedly, first on  $\mathcal{H}_v$ and then on $\mathcal{H}_w$, and so on. It can be shown \cite[Theorem 13.7]{Neumann} that this iterative projection  converges to the projection of the point onto  $\mathcal{H}_v\cap\mathcal{H}_w$. Three steps of this iterative process are depicted in Fig. \ref{figProj2}. We will now obtain the analytical expressions for projecting a point onto a hyperplane. Let $\mathcal{H}_v$
 be defined by
 \begin{align*}
\langle \textbf{s},\boldsymbol{\nu}^v\rangle\leq \beta_v.
 \end{align*}
 Let the point $\textbf{s}^{'}$ be such that
\begin{align*}
\beta_v^*\triangleq\langle\textbf{s}^{'},\boldsymbol{\nu}^v\rangle >\beta_v.
\end{align*}
Hence it lies outside $\mathcal{S}$. Let us define
\begin{align}
\textbf{s}^{''}=\textbf{s}^{'}-(\beta_v^*-\beta_v)\boldsymbol{\nu}^v. \label{ProjStepEqn}
\end{align}
Observe that $\textbf{s}^{''}$ is the orthogonal projection of $\textbf{s}^{'}$ onto $\mathcal{H}_v$.
\begin{align*}
			\langle \textbf{s}^{''}, \boldsymbol{\nu}^v \rangle &= \langle \textbf{s}^{'}, \boldsymbol{\nu}^v \rangle -(\beta_v^*-\beta_v)\langle \boldsymbol{\nu}^v,\boldsymbol{\nu}^v \rangle\\
 &= \beta_v^*  -(\beta_v^*-\beta_v)=\beta_v.
\end{align*} 
Since 	$\textbf{s}^{'}-\textbf{s}^{''}=(\beta_v^*-\beta_v)\boldsymbol{\nu}^v$, and  $\boldsymbol{\nu}^v$ is normal to the plane boundary of $\mathcal{H}_v$, it follows that the projection step (\ref{ProjStepEqn}) projects the point perpendicularly onto $\mathcal{H}_v$. We show below that the projection does not break any additional hyperplane constraints.
\newtheorem{prop}{Proposition}
\begin{prop}
	If $\langle\textbf{s}^{'},\boldsymbol{\nu}^w \rangle\leq \beta_w$, then  $\langle\textbf{s}^{''},\boldsymbol{\nu}^w \rangle \leq \beta_w$.
\end{prop}
\begin{proof}
	\begin{align*}
	\langle\textbf{s}^{''},\boldsymbol{\nu}^w \rangle &= \langle\textbf{s}^{'},\boldsymbol{\nu}^w\rangle -(\beta_v^*-\beta_v)\langle \boldsymbol{\nu}^v, \boldsymbol{\nu}^w\rangle.
	\end{align*}
	Since  $\boldsymbol{\nu}^w$ and $\boldsymbol{\nu}^v$ are non-negative, and $\beta_v^*>\beta_v$, we have $(\beta_v^*-\beta_v)\langle \boldsymbol{\nu}^v, \boldsymbol{\nu}^w\rangle\geq  0$, and consequently, $\langle\textbf{s}^{''},\boldsymbol{\nu}^w \rangle\leq \beta_w$.
\end{proof}	
Hence, if a point breaks just one hyperplane constraint and no other, the projection step (\ref{ProjStepEqn}) projects the point back on $\mathcal{S}$.

Consider an example. If the interference constraint is
\begin{align*}
s_1+s_2+s_4\leq 1,
\end{align*}
the projection step (\ref{ProjStepEqn}) is equivalent to 
\begin{align*}
s_1=s_1-\frac{s-1}{3},\ s_2=s_2-\frac{s-1}{3},\ s_4=s_4-\frac{s-1}{3}.
\end{align*}
where $s=s_1+s_2+s_4$.
\subsection{Convergence}
Let us define
\begin{align*}
f(\textbf{s}):=\sum_{k\in\mathcal{K}}f_k(\textbf{s}), \ f^*:=\max_{\textbf{s}\in\mathcal{S}}\sum_{k\in\mathcal{K}}f_k(\textbf{s}).
\end{align*}
We have the following theorem for the convergence of the distributed algorithm.
\newtheorem{theorem}{Theorem}
\begin{theorem}
	If $\max_{i,j,f}\theta^f\mu_{ij}\leq c_2$, 
	the algorithm defined by  the update equation (\ref{DistriEq}) followed by the projection (\ref{DistriEqP}) results in a sequence of points $\{\textbf{s}_n\}$ such that
	\begin{align*}
	\lim_{j\to\infty}\sup f(\textbf{s}_j) \geq f^*-c_3,
	\end{align*}
	where $c_3=\frac{\alpha\beta |\mathcal{K}|^2c_2^2}{2}$ with $\beta=4+\frac{1}{|\mathcal{K}|}$.
\end{theorem}
\begin{proof}
	See \cite{Ber10}.
\end{proof}
 We describe the algorithm below.
\subsection{Algorithm Description}
The algorithm proceeds in review cycles. At every slot $t$ that is the beginning of a review cycle, the nodes calculate the number of slots till the next review slot by 
\begin{align*}
T_{rev}=t+a_1\log (1+a_2\sum_{i,f}Q_i^f(t)),
\end{align*}
where $a_1$ and $a_2$ are constants. At the beginning of a review cycle, the nodes calculate the variables $\zeta_{ij}^f$  for all $i$, $j$ and $f$, and  use these till the end of the review cycle.  We will now describe how the $\zeta_{ij}^f$ variables are calculated at each node.
 
 The vector $\textbf{s}$ is initialized to all ones. The calculation proceeds cyclically.  The node which has the flow corresponding to the first component of the vector $\textbf{s}$ will do the update
 \begin{align}
 \textbf{s}(1)=\textbf{s}(1)+\alpha w(1)\mu(1).\label{nodeWiseUpdates}
 \end{align}
Here $w(1)=\theta^{f(1)} x_{i(1)}^{f(1)}$, with $\theta=1$ if the QoS constraint of flow $f(1)$ was satisfied in the previous review cycle; otherwise, it is set to be equal to a  value $\hat{\theta}$. The node then calculates the inner products
\begin{align*}
\beta_1^*\triangleq\langle \textbf{s},\boldsymbol{\nu}^1\rangle, \beta_2^*\triangleq\langle \textbf{s},\boldsymbol{\nu}^2\rangle
\end{align*} 
where $\boldsymbol{\nu}^l, \boldsymbol{\nu}^2$, correspond to the two interference constraints that the update step may break. If one of these constraints is broken, the update can be projected back in a single step. If both are violated, we will have to go for the iterative projection method. For projection on a plane characterized by $\langle \textbf{s},\boldsymbol{\nu}^i\rangle=\beta_i$, the node calculates $\beta_{ex}=\frac{\beta_i^*-\beta_i}{N_v}$ where $N_v$ is the number of links in that interference set. The node communicates this value to all links in its interference set. All these nodes, as well as the current node, update their values as
\begin{align*}
\textbf{s}(k)=\textbf{s}(k)-\beta_{ex}.
\end{align*} 
This is the projection step. Once the required number of projections is over, the node then passes its $s(1)$ to the node which has the next component of the vector $\textbf{s}$, and that node updates its value of $\textbf{s}(1)$. The next node now repeats the update and projection steps, and passes its update to its neighbour. This process is repeated cyclically, i.e, we repeat step (\ref{nodeWiseUpdates}) with 1 replaced by 2, and then by 3 and so on,  across the nodes till a predetermined stopping time is reached. At the end of the stopping time, we set all the negative components of $\textbf{s}$ to zero. For each interference set $I$, we check its constraint 
\begin{align*}
\langle\textbf{s},\boldsymbol{\nu}\rangle \leq \beta.
\end{align*}
If not, we apply the update
\begin{align*}
\textbf{s}(k)=\frac{\textbf{s}(k)}{\langle\textbf{s},\boldsymbol{\nu}\rangle}, \ k\in I.
\end{align*}
This will ensure compliance with the constraints. The complete algorithm is given below, as Algorithm \ref{DistriAlgo}, which uses in turn, Algorithms \ref{NodeWise}, \ref{NodeProject} and \ref{ScheCreate}. The last algorithm creates the schedule by scheduling flows on a link for a fraction of time equal to the corresponding $\textbf{s}(k)$.
\begin{algorithm}[h]
	\caption{Algorithm Q-Flo}
	\label{DistriAlgo}
	\begin{algorithmic}[1]
		\State $T_{rev}=0$, $T_{prev}=0$.
		
		\While{$t\geq 0$}
		
		\If{$t=T_{rev}$} 
		\State obtain variables $s_{ij}^f(T_{rev})$ using Algorithm \ref{NodeWise}

		\State $T_{prev}\gets T_{rev}$
		\State $T_{rev}\gets T_{rev}+a_1\log (1+a_2\sum_{i,f}Q_i^f(T_{rev}))$
		\State Create  $sched(i,j,f,t)$ from $t=T_{prev}$ to $t=T_{rev}-1$ using	Algorithm \ref{ScheCreate}	
		\EndIf
		\For{all $i,j,f$}
		\If{$Q_i^f(t)>\bar{q}_i^f$ and $sched(i,j,f,t)=1$} \ schedule flow $f$ across link $(i,j)$,  $\bar{q}_i^f=safety\ stock$
		\EndIf
		\EndFor
		
		\EndWhile
		
	\end{algorithmic}
	
\end{algorithm}

\begin{algorithm}[h]
	\caption{Algorithm at node level}
	\label{NodeWise}
	\begin{algorithmic}[1]
		\State Stopping time $T_s$, $t^{'}=0$, $s_{ij}^f(T_{rev})=0$ for all $i,j,f$
		\While{$t^{'}<T_s$}
		\State $k=t^{'}\Mod{|\mathcal{K}|}+1$, $(i,j,f)\gets (i(k),j(k),f(k))$
		\State If QoS criterion of  $f$ satisfied, $\theta^f\gets 2$; else $\theta^f\gets 1$
		\State $w\gets\theta^{f}Q_i^f(T_{rev})$, $\mu(k)\gets \mu_{ij}$, $s_{ij}^f\gets s_{ij}^f+\alpha w \mu(k)$
		\State Project $s_{ij}^f\gets \Pi_{\mathcal{S}}({s_{ij}^f})$ using Algorithm \ref{NodeProject}		
		
		\State $t^{'} \gets t^{'}+1$
		\EndWhile
		\State $s_{ij}^f\gets \max (s_{ij}^f,0)$
		\State If $s:=\sum_{j,f} s_{ij}^f+\sum_{j,f} s_{ji}^f>1$, $s_{ij}^f\gets \frac{s_{ij}^f}{s}$
		\State $s_{ij}^f(T_{rev})\gets s_{ij}^f$
		
	\end{algorithmic}

\end{algorithm}

\begin{algorithm}[h]
	\caption{Algorithm for Projection}
	\label{NodeProject}
	\begin{algorithmic}[1]
		\State Link interference constraints $\langle \textbf{s},\boldsymbol{\nu}^1\rangle\leq \beta_1, \langle \textbf{s},\boldsymbol{\nu}^2\rangle\leq \beta_2$
		\State Calculate  $\beta_1^*\triangleq\langle \textbf{s},\boldsymbol{\nu}^1\rangle, \beta_2^*\triangleq\langle \textbf{s},\boldsymbol{\nu}^2\rangle$
		\If{$\beta_i^*>\beta_i$} and {$\beta_j^*<\beta_j$}
		\State  $\beta_{ex}=\frac{\beta_i^*-\beta_i}{N_i+1}$, $N_i=$ number of interferers.
		\State For all interferers and current link, update $\beta_{ij}^f-\beta_{ex}$.
		\EndIf
		\If {$\beta_1^*>\beta_1$ and $\beta_2^*>\beta_2$}
		\State Repeat steps 4 to 6 and  8 to 11 $N\_rep$ times
		\EndIf
	\end{algorithmic}
	
\end{algorithm}

\begin{algorithm}[h]
	\caption{Algorithm for Schedule Creation}
	\label{ScheCreate}
	\begin{algorithmic}[1]
		\State Initialize $sched(i,j,f,t)=0\ \forall i,j,f,t$
		\For{$k \in \{1,\dots,|V|\}$}
		\State Obtain $sched(i,j,f,t)$ for $i\leq k-1$
		\State Obtain $s_{kj}^f(T_{rev})$ for all $j,f$
		\State Set of links that interfere with node $k =: N_k$
		\For{$j \in N_k, f\in F, t\in[T_{prev},T_{rev}]$}
		\If {$\sum_{i\leq k-1}sched(i,j,f,t)=0$ and $\sum_{i\in N_j}sched(j,i,f,t)=0$ and $\sum_{t^o=T_{prev}}^{t}sched(k,j,f,t^o)<s_{kj}^f(T_{rev}-T_{prev})$}
		\State $sched(k,j,f,t)=1$
		\EndIf
		\EndFor		
		\EndFor
	\end{algorithmic}					
\end{algorithm}
\section{Simulation Results}
We consider a 10 node network, with connectivity as depicted in Fig. \ref{figsim}, on a unit area, and Rayleigh distributed channel gains with parameter proportional to the inverse of the square of the distance between the nodes. The source-destination pairs are from node 0 to node 9, node 1 to node 7, node 5 to node 7, node 2 to node 8 and node 4 to node 9 with fixed routes being $0\to1\to3\to7\to9$, $0\to4\to9$ and  $0\to2\to6\to8\to9$ for the first flow, $1\to3\to7$ for the second, $5\to7$ for the third, $2\to6\to8$ for the fourth and $4\to9$ for the last. A packet is of size one bit. Nodes transmit with unit power.
\begin{figure}
	\begin{subfigure}[b]{0.24\textwidth}
		\centering
		\resizebox{\linewidth}{!}{
			\begin{tikzpicture}
			\node[draw,shape=circle, fill={rgb:orange,1;yellow,2;pink,5}, scale=0.6, transform shape] (v1) at (2,0) {$9$};
			\node[draw,shape=circle, fill={rgb:orange,1;yellow,2;pink,5}, scale=0.6, transform shape] (v2) at (4,-0.2) {$8$};
			\node[draw,shape=circle, fill={rgb:orange,1;yellow,2;pink,5}, scale=0.6, transform shape] (v3) at (3.8,0.7) {$6$};
			\node[draw,shape=circle, fill={rgb:orange,1;yellow,2;pink,5}, scale=0.6, transform shape] (v4) at (1,0.4) {$7$};
			\node[draw,shape=circle, fill={rgb:orange,1;yellow,2;pink,5}, scale=0.6, transform shape] (v5) at (2,1.4) {$3$};
			\node[draw,shape=circle, fill={rgb:orange,1;yellow,2;pink,5}, scale=0.6, transform shape] (v6) at (0,1.4) {$5$};
			\node[draw,shape=circle, fill={rgb:orange,1;yellow,2;pink,5}, scale=0.6, transform shape] (v7) at (1,2.4) {$1$};
			\node[draw,shape=circle, fill={rgb:orange,1;yellow,2;pink,5}, scale=0.6, transform shape] (v8) at (2,3.4) {$0$};
			\node[draw,shape=circle, fill={rgb:orange,1;yellow,2;pink,5}, scale=0.6, transform shape] (v9) at (3,1.4) {$4$};
			\node[draw,shape=circle, fill={rgb:orange,1;yellow,2;pink,5}, scale=0.6, transform shape] (v10) at (4,2.4) {$2$};
			\draw[line width=0.2mm] (v2) -- (v1)
			(v9) -- (v1)
			(v2) -- (v3)
			(v4) -- (v5)
			(v4) -- (v6)
			(v3) -- (v10)
			(v8) -- (v10)
			(v8) -- (v9)
			(v5) -- (v7)
			(v8) -- (v7)
			(v9) -- (v10)
			(v4) -- (v1);
			\end{tikzpicture}
		}
		\caption{Sample Network}
		\label{figsim}
	\end{subfigure}
	\begin{subfigure}[b]{0.24\textwidth}
		\centering
		\resizebox{\linewidth}{!}{
			\begin{tikzpicture}
			\begin{axis}[
			xlabel=Iterations,
			ylabel=Mean Delay(slots)]
			\addplot  coordinates {
				(1,427)
				(2,66)
				(3,52)
				(4,36)
				(5,28)
				(7,25)
				(10,23)
				(12,24)
				(15,49)
				(20,42)
			};
			
			\addplot coordinates {
				(1,563)
				(2,67)
				(3,38)
				(4,28)
				(5,22)
				(7,16)
				(10,13)
				(12,11)
				(15,26)
				(20,19)
			};
			
			\addplot coordinates {
				(1,104)
				(2,88)
				(3,64)
				(4,44)
				(5,33)
				(7,29)
				(10,27)
				(12,27)
				(15,35)
				(20,31)
			};
			
			\legend{Flow 7,Flow 8, Flow 9}
			\end{axis}
			\end{tikzpicture}
		}
		\caption{Number of Iterations versus Mean Delay}   
		\label{figIterVar}
	\end{subfigure}
	\caption{} 
	\label{fig:twoPics}
\end{figure}
We will first study the performance of the algorithm with the number of iterations of the distributed algorithm as parameter. We fix $\alpha=0.0001$, and the arrival process is Poisson with rate 3.3 corresponding to the flows from nodes 0 to 9, 1 to 7, 2 to 8, 4 to 9, and 5 to 7 respectively. The safety stock value is set to be 5 for all queues, and the simulation runs for $10^5$ slots. The constants $a_1$ and $a_2$ in Algorithm \ref{DistriAlgo} are set to 1.
	
In Fig. \ref{figIterVar} we plot the variation of mean delay of three flows in the network while we vary the number of iterations of the distributed algorithm. One iteration is equivalent to the completion of the update and project step at all the nodes. Since mean delay is directly proportional to mean queue length, it is evident that as the number of distributed iterations increases, the system has a lower mean queue length. From the simulations, around 5 rounds of iterations seem to be sufficient, and there is no major improvement in mean delay after that. There is, however, a marginal increase in the delay when the iterations increase further, to around 15. This is probably owing to the error accumulation as a result of the finite truncation of the iterative steps. Another parameter of interest is the number of rounds of iterative projection, $N\_{rep}$. From simulations, it seems that 2 to 4 rounds are sufficient.

We consider the case where we are trying to provide end-to-end mean delay guarantees to two flows: those destined to nodes 7 and 8 (Table \ref{table:One}). The arrival rate is 3.3 packets/slot for all arrivals. We study two cases, with $\hat{\theta}$ equal to 6 and 7. Using a higher weight, we are able to give tighter delay guarantees. However, a lower weight puts less strain on the other flows. Also, we see that as the delay constraint becomes tighter, the delay of the non QoS flow decreases. This is because while a given priority weight $\theta^f$ reserves resources for a QoS flow, if the delay required is smaller, the flow will have a smaller mean queue length, which will result in higher weight being given to non QoS flows in review periods where the delay criterion is satisfied, since the optimization function (\ref{optFun}) is  proportional to the queue length. Here $T_s=8$ and $N\_ rep=10$.\\
\begin{table}[h]
	\centering
	\caption{Two Flows with mean delay requirement}
	\label{table:One}
	\begin{tabular}{|p{0.68cm}|p{0.65cm}|p{0.65cm}|p{0.68cm}|p{0.65cm}|p{0.65cm}|p{0.62cm}|p{0.62cm}|}
		\hline
		\multicolumn{8}{|c|}{Mean Delay(slots)}\\		
		\hline 
		\multicolumn{3}{|c}{Flow 7}\ & \multicolumn{3}{|c}{Flow 8} & \multicolumn{2}{|c|}{Flow 9}\\
		\hline		
		\multicolumn{1}{|c}{Target}\ & \multicolumn{2}{|c}{Achieved}\ & \multicolumn{1}{|c}{Target} & \multicolumn{2}{|c}{Achieved} & \multicolumn{1}{|c|}{ Delay} & \multicolumn{1}{|c|}{ Delay}\\
		Mean Delay & with $\hat{\theta}=6$ & with $\hat{\theta}=7$ & Mean Delay & with $\hat{\theta}=6$ & with $\hat{\theta}=7$ & with $\hat{\theta}=6$ & with $\hat{\theta}=7$\\
		\hline
		50 & 51 & 51 & 30 & 32 &33 & 318 & 275 \\
		\hline
		40 & 40 & 40 & 25 & 26 & 28 & 253 & 196 \\
		\hline
		30 & 32 & 30 & 20 & 22 & 21 & 172 & 165 \\
		\hline
		25 & 30  & 26 & 15 & 18 & 15 & 145 & 147 \\
		\hline
	\end{tabular}
\end{table}
In Table \ref{table:Two}, we  demonstrate how to provide hard delay guarantee for flow 7 and mean delay guarantee for flow 8. The weights $\hat{\theta}$ for  flows 7 and 8 are 2 and 1.5. For flow 7, the packet is dropped at the destination if its deadline is not met. We have set a target of $2\%$ for such packets. We see that packets of flow 7 meet this target for the different deadlines fixed. The mean delay requirements of flow 8 are also met.
\begin{table}[h]
	\centering
	\caption{One mean delay and one hard deadline}
	\label{table:Two}
	\begin{tabular}{|p{1.4cm}|p{1.2cm}|p{1.2cm}|p{1.2cm}|p{1.2cm}|}
		\hline 
		\multicolumn{2}{|c}{Flow 7}\ & \multicolumn{2}{|c|}{Flow 8} & Flow 9\\
		\hline
		Hard Delay Target(slots), Drop Ratio Target & Drop Ratio Achieved & Mean Delay Target (slots) & Mean Delay Achieved (slots) & Mean Delay (slots)\\
		\hline
		180,2\% & 2\% & 50 & 51 & 136 \\
		\hline
		180,2\% & 2\% & 40 & 43 & 100 \\
		\hline
		180,2\% & 2\% & 35 & 36 & 89 \\
		\hline
		160,2\% & 2\% & 45 & 45 &  88 \\
		\hline
		140,2\% & 2\% & 30 & 33 & 91 \\
		\hline
		120,2\% & 2\% & 35 & 37 &  94 \\
		\hline
	\end{tabular}
\end{table}
\section{Conclusion}
We have developed a distributed algorithm to provide Quality-of-Service requirements in terms of end-to-end mean delay guarantees and hard deadline guarantees to flows in a multihop wireless network. The algorithm optimizes, in a distributed fashion, a function with distributed weights given to pseudo draining times, with the weights  varied dynamically  to provide priority for flows in the network, and consequently, meeting their respective delay constraints. We use iterative gradient ascent and distributed iterative projection methods in order to compute the optimal point in a distributed manner. By means of simulations we establish the efficacy of the algorithm in providing the required delay demands. We also see via simulations that the algorithm converges quickly.
\bibliography{survey}
\bibliographystyle{IEEEtran}

\end{document}